\documentclass[journal, draftcls, onecolumn]{IEEEtran}\usepackage{amssymb,amsmath,graphicx,cite,graphics,mathrsfs, mathtools}
\usepackage{latexsym,amssymb,amsmath,graphicx,cite,bbm,float,graphics}
\usepackage{algorithmic, amsthm, xcolor, url}
\usepackage[ruled,vlined, linesnumbered]{algorithm2e}
\usepackage{threeparttable}
\usepackage{multirow}
\usepackage{etoolbox}
\makeatletter
\patchcmd{\@makecaption}
  {\scshape}
  {}
  {}
  {}
\makeatother

\newtheorem{proposition}{Proposition}
\newtheorem{theorem}{Theorem}
\newtheorem{lemma}{Lemma}
\newtheorem{definition}{Definition}
\newtheorem{claim}{Claim}
\newcommand\given[1][]{\:#1\vert\:}

\usepackage{xcolor}

\title{On the Optimality of the Kautz-Singleton Construction in Probabilistic Group Testing}

\author{Huseyin A. Inan$^{1}$, Peter Kairouz$^{1}$, Mary Wootters$^{2}$, and Ayfer Ozgur$^{1}$
\thanks{$^{1}$Department of Electrical Engineering, Stanford University.
        {\tt\small hinan1, kairouzp, aozgur@stanford.edu}}%
\thanks{$^{2}$Departments of Computer Science and Electrical Engineering,
Stanford University.
        {\tt\small marykw@stanford.edu}}%
        \thanks{This work appeared in part at Allerton 2018.}%
}

\begin{document}

\maketitle

\begin{abstract}
We consider the probabilistic group testing problem where $d$ random defective items in a large population of $N$ items are identified with high probability by applying binary tests. It is known that $\Theta(d \log N)$ tests are necessary and sufficient to recover the defective set with vanishing probability of error when $d = O(N^{\alpha})$ for some $\alpha \in (0, 1)$. However, to the best of our knowledge, there is no explicit (deterministic) construction achieving $\Theta(d \log N)$ tests in general. In this work, we show that a famous construction introduced by Kautz and Singleton for the combinatorial group testing problem (which is known to be suboptimal for combinatorial group testing for moderate values of $d$) achieves the order optimal $\Theta(d \log N)$ tests in the probabilistic group testing problem when $d = \Omega(\log^2 N)$. This provides a strongly explicit construction achieving the order optimal result in the probabilistic group testing setting for a wide range of values of $d$. To prove the order-optimality of Kautz and Singleton's construction in the probabilistic setting, we provide a novel analysis of the probability of a non-defective item being covered by a random defective set directly, rather than arguing from combinatorial properties of the underlying code, which has been the main approach in the literature. Furthermore, we use a recursive technique to convert this construction into one that can also be efficiently decoded with only a log-log factor increase in the number of tests.

\end{abstract}

\section{Introduction}
 The objective of group testing is to efficiently identify a small set of $d$ defective items in a large population of size $N$ by performing binary tests on groups of items, as opposed to testing the items individually. A positive test outcome indicates that the group contains at least one defective item. A negative test outcome indicates that all the items in the group are non-defective. When $d$ is much smaller than $N$, the defectives can be identified with far fewer than $N$ tests.

The original group testing framework was developed in 1943 by Robert Dorfman \cite{dorfman1943detection}. Back then, group testing was devised to identify which WWII draftees were infected with syphilis, without having to test them individually. In Dorfman's application, items represented draftees and tests represented actual blood tests. Over the years, group testing has found numerous applications in fields spanning biology \cite{chen2008decoding}, medicine \cite{ganesan2017learning}, machine learning \cite{malioutov2013exact}, data analysis \cite{gilbert2008group}, signal processing \cite{emad2014poisson}, and wireless multiple-access communications \cite{Berger1984,wolf1985born,luo2008neighbor,Fletcher2009}.

\subsection{Non-adaptive probabilistic group testing}
Group testing strategies can be \emph{adaptive}, where the $i^{th}$ test is a function of the outcomes of the $i-1$ previous tests, or \emph{non-adaptive}, where all tests are designed in one shot. A non-adaptive group testing strategy can be represented by a $t \times N$ binary matrix ${M}$, where $M_{ij} = 1$ indicates that item $j$ participates in test $i$. Group testing schemes can also be \emph{combinatorial} \cite{ngo2000survey,du2000combinatorial} or \emph{probabilistic} \cite{atia2009, sejdinovic2010, chan2001probabilisticgt,  cevher16, johnson14, scarlett18, scarlett_cevher_17, johnson_noisy_18}. 

The goal of combinatorial group testing schemes is to recover any set of up to $d$ defective items with zero-error and require at least $t = \min \{N, \Omega(d^2 \log_d N) \}$ tests \cite{d1982bounds, furedi1996onr}. A strongly explicit construction\footnote{We will call a $t \times N$ matrix \em strongly explicit \em  if any column of
the matrix can be constructed in time $\textnormal{poly}(t)$. A matrix will be called \em explicit \em if it
can be constructed in time $\textnormal{poly}(t, N)$.} that achieves $t = O(d^2 \log_d^2 N)$ was introduced by Kautz and Singleton in \cite{kautz1964}. A more recent explicit construction achieving $t = O(d^2 \log N)$ was introduced by Porat and Rothschild in \cite{porat2008}. We note that the Kautz-Singleton construction 
matches the best known lower bound $\Omega(d^2 \log_d N)$ in the regime where $d = \Theta(N^{\alpha})$ for some $\alpha \in (0, 1)$. However, for moderate values of $d$ (e.g., $d = O(\textnormal{poly}(\log N))$), the construction introduced by Porat and Rothschild achieving $t = O(d^2 \log N)$ is more efficient and the Kautz-Singleton construction is suboptimal in this regime. 

In contrast, probabilistic group testing schemes assume a random defective set of size $d$, allow for an arbitrarily small probability of reconstruction error, and require only $t =\Theta( d \log N)$ tests  when $d = O(N^{1-\alpha})$ for some $\alpha \in (0, 1)$ \cite{chan2001probabilisticgt, cevher16, johnson14}. In this paper, we are interested in non-adaptive probabilistic group testing schemes.

\subsection{Our contributions }
To best of our knowledge, all known probabilistic group testing strategies that achieve $t = O(d \log N)$ tests are randomized (i.e., $M$ is randomly constructed) \cite{atia2009, sejdinovic2010, chan2001probabilisticgt,  cevher16, johnson14, scarlett18, scarlett_cevher_17, johnson_noisy_18}. Recently, Mazumdar \cite{mazumdar16} presented explicit schemes (deterministic constructions of $M$) for probabilistic group testing framework. This was done by studying the average and minimum Hamming distances of constant-weight codes (such as Algebraic-Geometric codes) and relating them to the properties of group testing strategies. However, the explicit schemes in \cite{mazumdar16} achieve $t = \Theta(d \log^2 N/\log d)$, which is not order-optimal when $d$ is poly-logarithmic in $N$. It is therefore of interest to find explicit, deterministic schemes that achieve $t = O(d\log N)$ tests.

This paper presents a strongly explicit scheme that achieves $t = O(d\log N)$ in the regime where $d = \Omega(\log^2 N)$. We show that Kautz and Singleton's well-known scheme is order-optimal for probabilistic group testing. This is perhaps surprising because for moderate values of $d$ (e.g., $d = O(\textnormal{poly}(\log N))$), this scheme is known to be sub-optimal for combinatorial group testing. We prove this result for both the noiseless and noisy (where test outcomes can be flipped at random) settings of probabilistic group testing framework. We prove the order-optimality of Kautz and Singleton's construction by analyzing the probability of a non-defective item being ``covered'' (c.f. Section \ref{sysmod}) by a random defective set directly, rather than arguing from combinatorial properties of the underlying code, which has been the main approach in the literature  \cite{kautz1964, porat2008, mazumdar16}. 

We say a group testing scheme, which consists of a group testing strategy (i.e., $M$) and a decoding rule, achieves probability of error $\epsilon$ and is \em efficiently decodable \em if the decoding rule can identify the defective set in $\textnormal{poly}(t)$-time complexity with $\epsilon$ probability of error. While we can achieve the decoding complexity of $O(t N)$ with the ``cover decoder'' (c.f. Section \ref{sysmod})\footnote{Common constructions in group testing literature have density $\Theta(1/d)$, therefore, the decoding complexity can be brought to $O(tN/d)$.}, our goal is to bring the decoding complexity to $\textnormal{poly}(t)$. To this end, we use a recursive technique inspired by \cite{ngo11} to convert the Kautz-Singleton construction into a strongly explicit construction with $t = O(d \log N \log \log_d N)$ tests and decoding complexity $O(d^3 \log N  \log \log_d N)$. This provides an efficiently decodable scheme with only a log-log factor increase in the number of tests. Searching for order-optimal explicit or randomized constructions that are efficiently decodable remains an open problem.


\subsection{Outline }
The remainder of this paper is organized as follows. In Section \ref{sysmod}, we present the system model and necessary prerequisites. The optimality of the Kautz-Singleton construction in the probabilistic group testing setting is formally presented in Section \ref{sec:mainres}. We propose an efficiently decodable group testing strategy in Section \ref{sec:decoding}. We defer the proofs of the results to their corresponding sections in the appendix. We provide, in Section \ref{sec:related}, a brief survey of related results on  group testing and a detailed comparison with Mazumdar's recent work in \cite{mazumdar16}. Finally, we conclude our paper in Section \ref{sec:conclusion} with a few interesting open problems.

\section{System Model and Basic Definitions}
\label{sysmod}
For any $t \times N$ matrix ${M}$, we use ${M}_i$ to refer to its $i$'th column and ${M}_{ij}$ to refer to its $(i, j)$'th entry. The \textit{support} of a column $M_i$ is the set of coordinates where $M_i$ has nonzero entries. For an integer $m \geq 1$, we denote the set $\{1, \ldots, m\}$ by $[m]$. The Hamming weight of a column of $M$ will be simply referred to as the \textit{weight} of the column.

We consider a model where there is a random defective set $S$ of size $d$ among the items
$[N]$. We define $\mathcal{S}$ as the set of all possible defective sets, i.e., the set of $\binom{N}{d}$ subsets of $[N]$ of cardinality $d$ and we let $S$ be uniformly distributed over $\mathcal{S}$.\footnote{This assumption is not critical. Our results carry over to the setting where the defective items are sampled with replacement.} The goal is to determine $S$ from the binary measurement vector $Y$ of size $t$ taking the form
\begin{align}
Y = \left(\bigvee_{i \in S} M_i \right)\oplus v,
\label{model}
\end{align}
where $t \times N$ measurement matrix $M$ indicates which items are included in the test, i.e., ${M}_{ij} = 1$ if the item $j$ is participated in test $i$, $v \in \{0,1\}^t$ is a noise term, and $\oplus$ denotes modulo-2 addition. In words, the measurement vector $Y$ is the Boolean OR combination of the columns of the measurement matrix $M$ corresponding to the defective items in a possible noisy fashion. We are interested in both the noiseless and noisy variants of the model in \eqref{model}. In the noiseless case, we simply consider $v = 0$, i.e., $Y = \bigvee_{i \in S} M_i $. Note that the randomness in the measurement vector $Y$ is only due to the random defective set in this case. On the other hand, in the noisy case we consider $v \sim \textnormal{Bernoulli}(p)$ for some fixed constant $p \in (0, 0.5)$, i.e., each measurement is independently flipped with probability~$p$.

Given $M$ and $Y$, a decoding procedure forms an estimate $\hat{S}$ of $S$. The performance measure we consider in this paper is the \textit{exact recovery} where the average probability of error is given by
\begin{align*}
P_e \triangleq \Pr (\hat{S} \neq S),
\end{align*}
and is taken over the realizations of $S$ and $v$ (in the noisy case). The goal is to minimize the total number of tests $t$ while achieving a vanishing probability of error, i.e., satisfying $P_e \rightarrow 0$.

\subsection{Disjunctiveness}

We say that a column $M_i$ is \emph{covered} by a set of columns $M_{j_1}, \ldots, M_{j_l}$ with ${j_1}, \ldots, {j_l} \in [N]$ if the support of $M_i$ is contained in the union of the supports of columns  $M_{j_1}, \ldots, M_{j_l}$. A binary matrix $M$ is called $d$-\textit{disjunct} if any column of $M$ is not covered by any other $d$ columns. The $d$-disjunctiveness property ensures that we can recover any defective set of size $d$ with zero error from the measurement vector $Y$ in the noiseless case. This can be naively done using the \textit{cover decoder} (also referred as the COMP decoder \cite{chan2001probabilisticgt, johnson14}) which runs in $O(t N)$-time. The cover decoder simply scans through the columns of $M$, and returns the ones that are covered by the measurement vector $Y$.
When $M$ is $d$-disjunct, the cover decoder succeeds at identifying all the defective items without any error.

In this work, we are interested in the probabilistic group testing problem where the 0-error condition is relaxed into a vanishing probability of error. Therefore we can relax the $d$-disjunctiveness property. Note that to achieve an arbitrary but fixed $\epsilon$ average probability of error in the noiseless case, it is sufficient to ensure that at least $(1-\epsilon)$ fraction of all possible defective sets do not cover any other column. A binary matrix satisfying this relaxed form is called an \textit{almost disjunct} matrix \cite{mazumdar16, macula04, malyutov78, Zhigljavsky03} and with this condition one can achieve the desired $\epsilon$ average probability of error by applying the cover decoder.
\subsection{Kautz-Singleton Construction}
\label{sec:KS_cons}

In their work \cite{kautz1964}, Kautz and Singleton provide a construction of disjunct matrices by converting a Reed-Solomon (RS) code \cite{RS} to a binary matrix. We begin with the definition of Reed-Solomon codes.
\begin{definition}
Let $\mathbb{F}_q$ be a finite field and $\alpha_1, \ldots, \alpha_n$ be distinct elements from $\mathbb{F}_q$. Let $k \leq n \leq q$. The Reed-Solomon code of dimension $k$ over $\mathbb{F}_q$, with evaluation points $\alpha_1, \ldots, \alpha_n$ is defined with the following encoding function. The encoding of a message $m = (m_0, \ldots, m_{k-1})$ is the evaluation of the corresponding $k-1$ degree polynomial $f_m(X) = \sum_{i=0}^{k-1} m_i X^i$ at all the
$\alpha_i$'s:
\begin{align*}
\textnormal{RS}(m) = (f_m(\alpha_1), \ldots, f_m(\alpha_n)).
\end{align*}
\label{RS}
\end{definition}

The Kautz-Singleton construction starts with a $[n, k]_q$ RS code with $n = q$ and $N = q^k$. Each $q$-ary symbol is then replaced by a unit weight binary vector of length $q$, via ``identity mapping" which takes a symbol $i \in [q]$ and maps it to the vector in $\{0, 1\}^q$ that has a 1 in the $i$'th position and zero everywhere else. Note that the resulting binary matrix will have $t = n q = q^2$ tests. An example illustrating the Kautz-Singleton construction is depicted in Figure \ref{fig:ks}. This construction achieves a $d$-disjunct $t \times N$ binary matrix with $t = O(d^2 \log_d^2 N)$ by choosing the parameter $q$ appropriately. While the choice $n=q$ is appropriate for the combinatorial group testing problem, we will shortly see that we need to set $n = \Theta(\log N)$ in order to achieve the order-optimal result in the probabilistic group testing problem.

\begin{figure}
  \centering
    \includegraphics[width=0.7\textwidth]{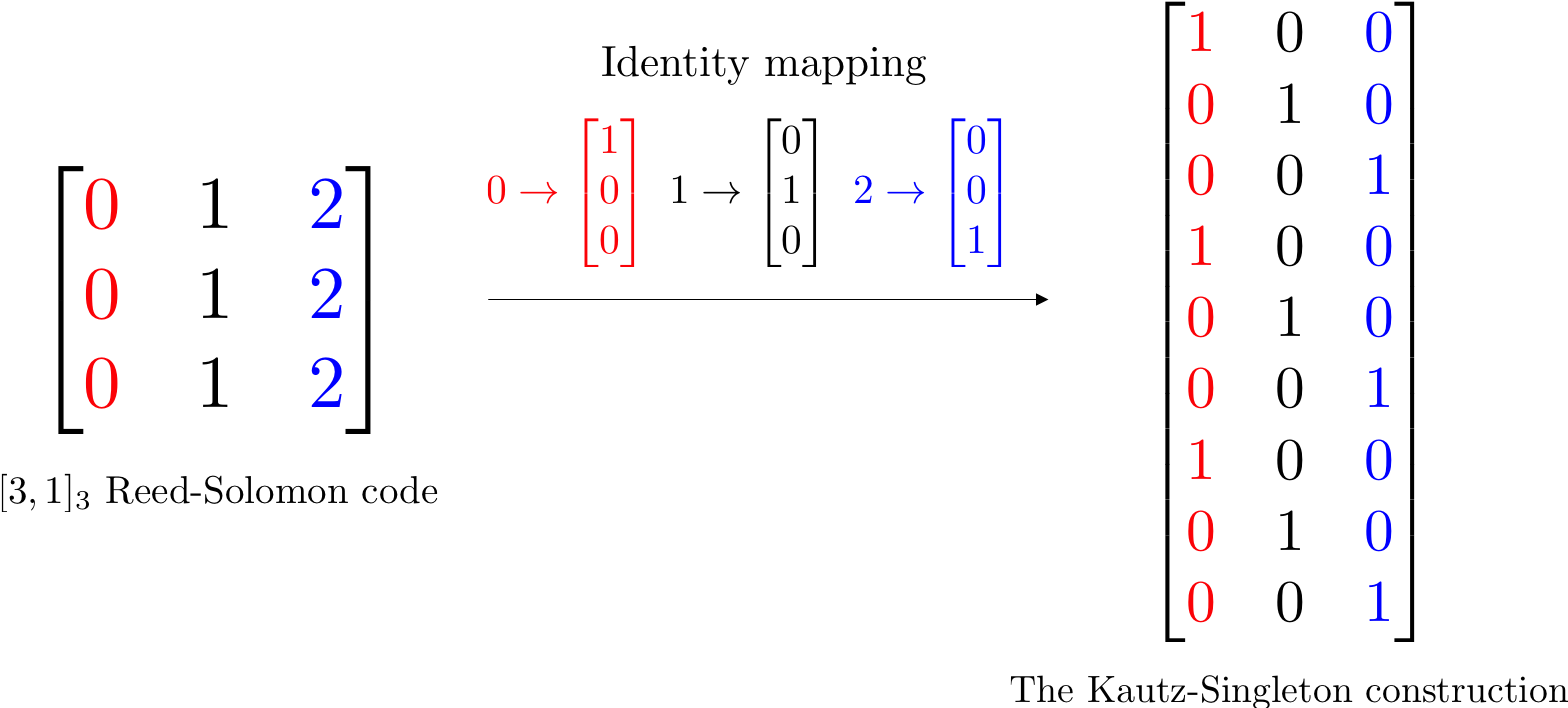}
      \caption{An example of the Kautz-Singleton construction with $[3, 1]_3$ Reed-Solomon code.}
      \label{fig:ks}
\end{figure}

While this is a strongly explicit construction, it is suboptimal for combinatorial group testing in the regime $d = O(\textnormal{poly}(\log N))$: an explicit construction with smaller $t$ (achieving $t = O(d^2 \log N)$) is introduced by Porat and Rothschild in \cite{porat2008}. 
Interestingly, we will show in the next section that
this same strongly explicit construction that is suboptimal for combintorial group testing in fact achieves the order-optimal $t = \Theta(d \log N)$ result in both the noiseless and noisy versions of probabilistic group testing.

\section{Optimality of the Kautz-Singleton construction}
\label{sec:mainres}

We begin with the noiseless case ($v=0$ in \eqref{model}). The next theorem shows the optimality of the Kautz-Singleton construction with properly chosen parameters $n$ and $q$. 
\begin{theorem}
Let $\delta > 0$. Under the noiseless model introduced in Section \ref{sysmod}, the Kautz-Singleton construction with parameters $q = c_1 d$ for any $c_1 \geq 4$ and $n = (1 + \delta) \log N$ has average probability of error $P_e \leq N^{- \Omega(\log q)} + N^{-\delta}$ under the cover decoder in the regime $d = \Omega(\log^2 N)$.
\label{thm:Exp_noiseless}
\end{theorem}
The proof of the above theorem can be found in Appendix \ref{appendix:noiseless}. We note that the Kautz-Singleton construction in Theorem \ref{thm:Exp_noiseless} has $t = n q = \Theta(d \log N)$ tests, therefore, achieving the order-optimal result in the probabilistic group testing problem in the noiseless case. It is further possible to extend this result to the noisy setting where we consider $v \sim \textnormal{Bernoulli}(p)$ for some fixed constant $p \in (0, 0.5)$, i.e., each measurement is independently flipped with probability $p$. Our next theorem shows the optimality of the Kautz-Singleton construction in this case.
\begin{theorem}
Let $\delta > 0$. Under the noisy model introduced in Section \ref{sysmod} with some fixed noise parameter $p \in (0, 0.5)$, the Kautz-Singleton construction with parameters $q = c_1 d$ for any $c_1 \geq 24$ and $n = c_2(1 + \delta) \log N$ for any $c_2 \geq \max \{\frac{8}{9(0.5-p)^2}, 40.57\}$ has average probability of error $P_e \leq N^{- \Omega(\log q)} + 3N^{-\delta}$ under the modified version of cover decoder in the regime $d = \Omega(\log^2 N)$.
\label{thm:Exp_noisy}
\end{theorem}
The proof of the above theorem can be found in Appendix \ref{appendix:noisy}. Similar to the noiseless setting, the Kautz-Singleton construction provides a strongly explicit construction achieving optimal number of tests $t = \Theta(d \log N)$ in the noisy case.

Given that the Kautz-Singleton construction achieves a vanishing probability of error with $t = \Theta(d \log N)$ order-optimal number of tests, a natural question of interest is how large the constant is and how the performance of this construction compares to random designs for given finite values of $d$ and $N$. To illustrate the empirical performance of the Kautz-Singleton construction in the noiseless case, we provide simulation results in Figure \ref{fig:sim1} and \ref{fig:sim2} for different choices of $N$ and $d$ and compare the results to random designs considered in the literature. We used the code in \cite{Erlich15_code} (see \cite{Erlich15} for the associated article) for the Kautz-Singleton construction. For comparison, we take two randomized constructions from the literature, namely
the Bernoulli design (see \cite{johnson14}) and the near-constant column weight design studied in \cite{scarlett18}. We use the cover decoder for decoding. 
The simulation results illustrate that the Kautz-Singleton construction achieves better success probability for the same number of tests, which suggests that the implied constant for the Kautz-Singleton construction may be better than those for these random designs; we note that similar empirical findings were observed in \cite{Erlich15}. Since the Kautz-Singleton construction additionally has an explicit and simple structure, this construction may be a good choice for designing measurement matrices for probabilistic group testing in practice.

\begin{figure}
  \centering
    \includegraphics[width=0.7\textwidth]{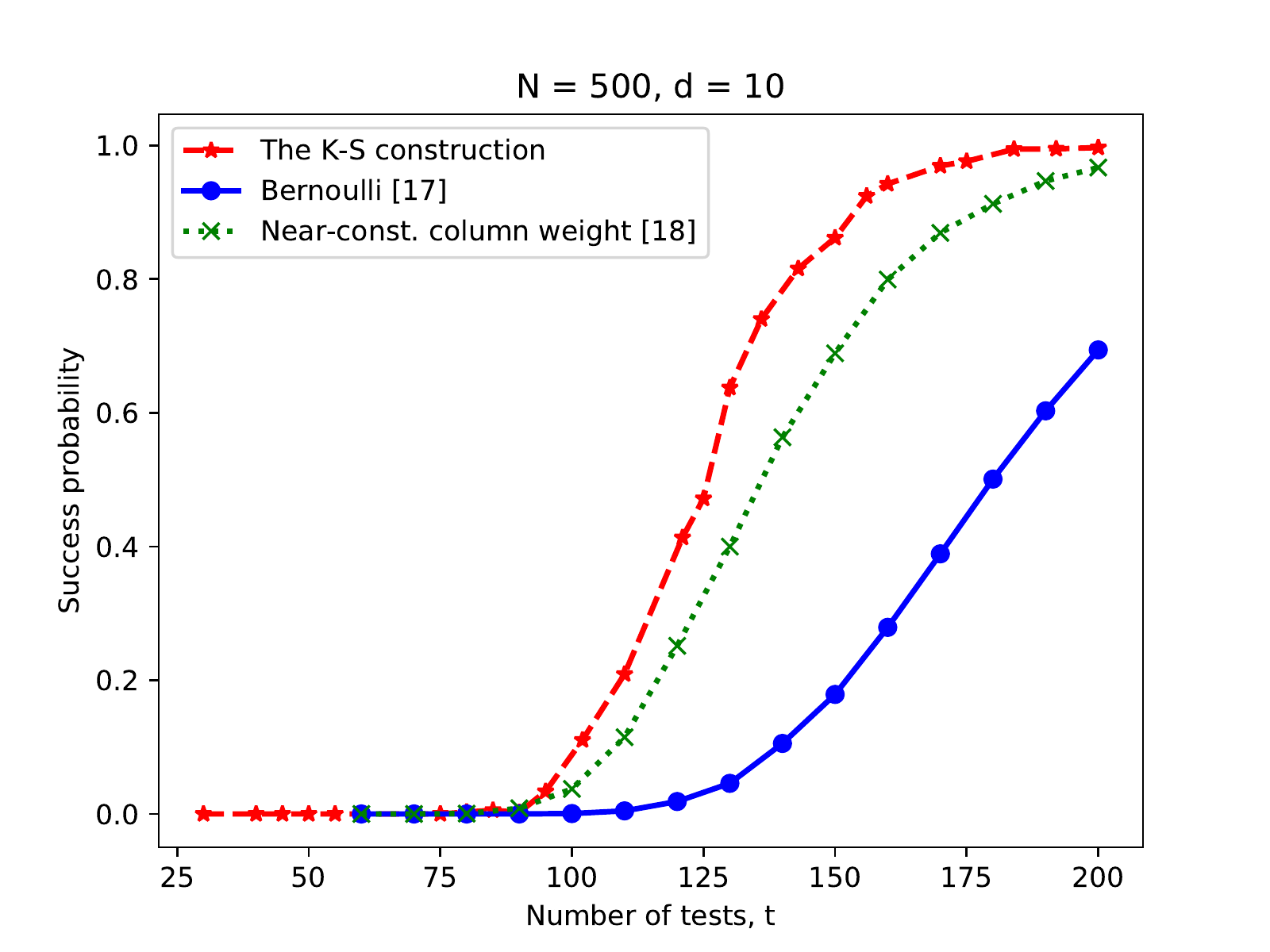}
      \caption{ Empirical performances of the Kautz-Singleton construction along with the random near-constant column weight \cite{scarlett18} and Bernoulli designs \cite{johnson14} under the cover decoder for $N = 500$ items and $d = 10$ defectives. For the Kautz-Singleton construction, empirical performance was judged using 5000 random trials and the number of tests correspond to a range of $(q, n)$ pair selections. For the random matrices, empirical performance was judged from 100 trials each on 100 random matrices.}
      \label{fig:sim1}
\end{figure}

\begin{figure}
  \centering
    \includegraphics[width=0.7\textwidth]{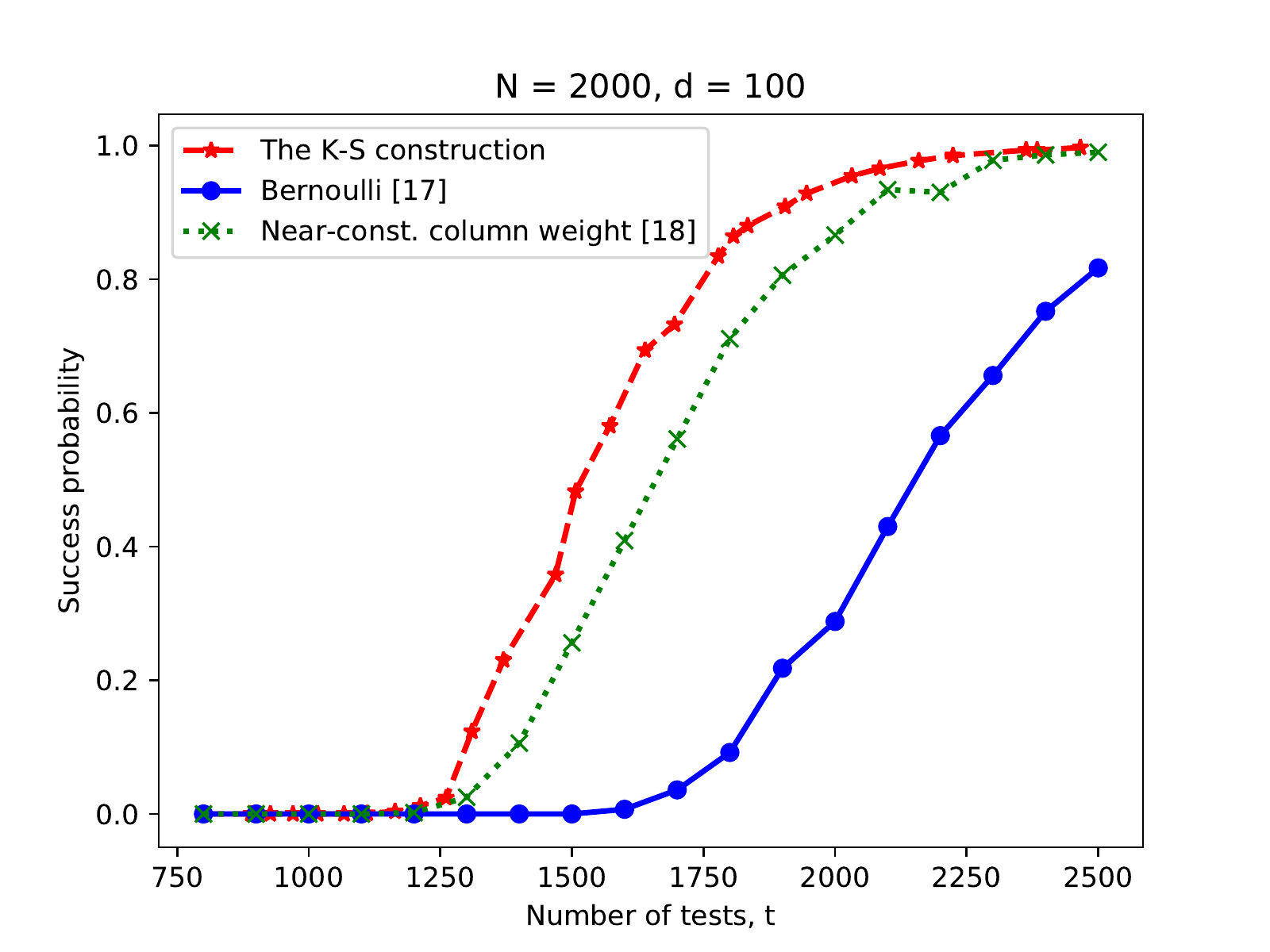}
      \caption{ Empirical performances of the Kautz-Singleton construction along with the random near-constant column weight \cite{scarlett18} and Bernoulli designs \cite{johnson14} under the cover decoder for $N = 2000$ items and $d = 100$ defectives. For the Kautz-Singleton construction, empirical performance was judged using 5000 random trials and the number of tests correspond to a range of $(q, n)$ pair selections. For the random matrices, empirical performance was judged from 100 trials each on 100 random matrices.}
      \label{fig:sim2}
\end{figure}

\section{Decoding}
\label{sec:decoding}

While the cover decoder, which has a decoding complexity of $O(t N)$, might be reasonable for certain applications, there is a recent research effort towards low-complexity decoding schemes due to the emerging applications involving massive datasets \cite{mahdi2009, indyk2010, ngo11, lee2016saffron, jaggi2017}. The target is a decoding complexity of $\textnormal{poly}(t)$. This is an exponential improvement in the running time over the cover decoder for moderate values of $d$. For the model we consider in this work (i.e., exact recovery of the defective set with vanishing probability of error), there is no known efficiently decodable scheme with optimal $t = \Theta(d \log N)$ tests to the best of our knowledge. The work \cite{lee2016saffron} presented a randomized scheme which identifies all the defective items with high probability with $O(d \log d \log N)$ tests and time complexity $O(d \log d \log N)$.
Another recent result, \cite{jaggi2017}, introduced an algorithm which requires $O(d \log d \log N)$ tests with $O(d(\log^2 d + \log N))$ decoding complexity. Note that the decoding complexity reduces to $O(d \log N)$ when $d = O(\textnormal{poly}(\log N))$ which is order-optimal (and sub-linear in the number of tests), although the number of tests is not. In both \cite{lee2016saffron} and \cite{jaggi2017}, the number of tests is away from the optimal number of tests by a factor of $\log d$. 

We can convert the strongly explicit constructions in Theorem \ref{thm:Exp_noiseless} and \ref{thm:Exp_noisy} into strongly explicit constructions that are also efficiently decodable by using a recursive technique introduced in \cite{ngo11} where the authors construct efficiently decodable error-tolerant list disjunct matrices. For the sake of completeness, we next discuss the main idea applied to our case.

The cover decoder goes through the columns of $M$ and decides whether the corresponding item is defective or not. This results in  decoding complexity $O(t N)$. Assume we were given a superset $S'$ such that $S'$ is guaranteed to include the defective set $S$, i.e. $S \subseteq S'$, then the cover decoder could run in time $O(t \cdot \vert S' \vert)$ over the columns corresponding to $S'$, which  depending on the size of $S'$ could result in significantly lower complexity. It turns out that we can construct this small set $S'$ recursively.

Suppose that we have access to an efficiently decodable $t_1(d, \sqrt{N}, \epsilon/4, p) \times \sqrt{N}$ matrix $M^{(1)}$ which can be used to detect at most $d$ defectives among $\sqrt{N}$ items with probability of error $P_e \leq {\epsilon/4}$ when the noise parameter is $p$ by using $t_1({d}, \sqrt{N}, {\epsilon/4}, {p})$ tests. We construct two $t_1(d, \sqrt{N}, \epsilon/4, p) \times N$ matrices $M^{(F)}$ and $M^{(L)}$ using $M^{(1)}$ as follows. For $i \in [N]$, the $i$'th column of $M^{(F)}$ is equal to  $j$'th column of $M^{(1)}$ if the \textit{first} $\frac{1}{2} \log N$ bits in the binary representation of $i$ are given by the binary representation of $j$ for $j \in [\sqrt{N}]$. Similarly, for $i \in [N]$, the $i$'th columns of $M^{(L)}$ is equal to the $j$'th column of $M^{(1)}$ if the \textit{last} $\frac{1}{2} \log N$ bits in the binary representation of $i$ are given by the binary representation of $j$ for $j \in [\sqrt{N}]$.

The final matrix matrix $M$ is constructed by vertically stacking $M^{(F)}$, $M^{(L)}$ and a $t_2(d, N, \epsilon/2, p) \times N$ matrix $M^{(2)}$ which is not necessarily efficiently decodable (e.g., the Kautz-Singleton construction). As before, $t_2(d, N, \epsilon/2, p)$ is the number of tests for $M^{(2)}$, which we assume can be used to detect  $d$ defectives among $N$ items with probability of error $P_e \leq {\epsilon/2}$  when the noise parameter is $p$. The decoding works as follows. We obtain the measurement vectors $Y^{(F)}$, $Y^{(L)}$, and $Y^{(2)}$ given by $Y^{(F)} = \bigvee_{i \in S} M^{(F)}_{i} \oplus v^{(F)}$, $Y^{(L)} = \bigvee_{i \in S} M^{(L)}_{i} \oplus v^{(L)}$, and $Y^{(2)} = \bigvee_{i \in S} M^{(2)}_{i} \oplus v^{(2)}$ respectively where $v^{(F)}$, $v^{(L)}$, and $v^{(2)}$ are the noise terms corrupting the corresponding measurements. We next apply the decoding algorithm for $M^{(1)}$ to $Y^{(F)}$ and $Y^{(L)}$ to obtain the estimate sets $\hat{S}^{(F)}$ and $\hat{S}^{(L)}$ respectively. Note that the sets $\hat{S}^{(F)}$ and $\hat{S}^{(L)}$ can decode the first and last $\frac{1}{2} \log N$-bits of the defective items respectively with probability at least $1 - \epsilon/2$ by the union bound. Therefore, we can construct the set $S' = \hat{S}^{(F)} \times \hat{S}^{(L)}$ where $\times$ denotes the Cartesian product and obtain a super set that contains the defective set $S$ with error probability at most $\epsilon/2$. We further note that since $|\hat{S}^{(F)}| \leq d$ and $|\hat{S}^{(L)}| \leq d$, we have $\vert S' \vert \leq d^2$.
We finally apply the naive cover decoder to $M^{(2)}$ by running it over the set $S'$ to compute the final estimate $\hat{S}$ which can be done in additional $O(t_2 \cdot d^2)$ time. Note that by the union bound the probability of error is bounded by $\epsilon$. Figure \ref{fig:dec} illustrates the main idea with the example of $d=2$ and $N=16$.
We provide this decoding scheme in Algorithm \ref{KSdec} for the special case $N = d^{2^{i}}$ for some non-negative integer $i$ although the results hold in the general case and no extra assumption beyond $d = \Omega(\log^2 N)$ is needed. The next theorem is the result of applying this idea recursively.

\begin{figure}
  \centering
    \includegraphics[width=1\textwidth]{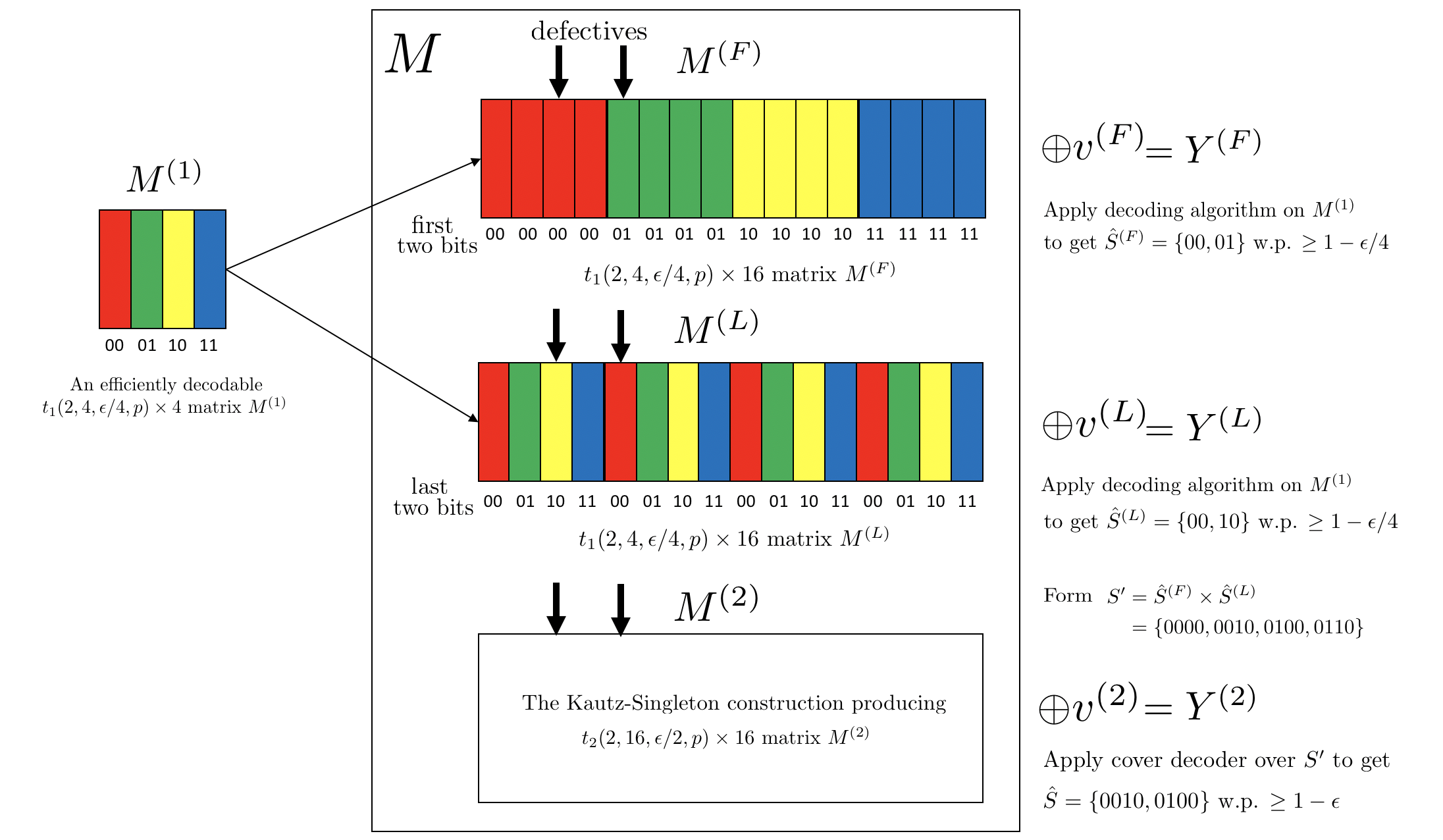}
      \caption{An illustration of the construction presented in Section \ref{sec:decoding} for the case $d=2$ and $N=16$. The illustration depicts the main idea, and the full construction is achieved by applying this idea recursively.}
      \label{fig:dec}
\end{figure}

\begin{algorithm}
\KwIn{The measurement vector $Y$, the group testing matrix $M$, the defective set size $d$, the number of items $N$}
\KwOut{The defective set estimate $\hat{S}$}
    \eIf{$N=d$}
      {
        Return the defective set using $Y$ (individual testing)\;
      }
      {
        Compute $M^{(1)}$ and $M^{(2)}$ (as described in the text)\;
        Compute $Y^{(F)}$ and $Y^{(L)}$  (as described in the text)\;
        $\hat{S}^{(F)} = decode(Y^{(F)}, M^{(1)}, d, \sqrt{N})$\;
        $\hat{S}^{(L)} = decode(Y^{(L)}, M^{(1)}, d, \sqrt{N})$\;
        \If{$\vert \hat{S}^{(F)} \vert > d$ or $\vert \hat{S}^{(L)} \vert > d$}{
        	return \{\}\;
        }
        Construct $S'= \hat{S}^{(F)} \times \hat{S}^{(L)}$\;
        Apply the cover decoder to $M^{(2)}$ over the set $S'$ and compute $\hat{S}$\;
        Return $\hat{S}$\;
      }
    \caption{The decoding alg. \textit{decode}($Y$, $M$, $d$, $N$)}
\label{KSdec}    
\end{algorithm}

\vspace{20pt}

\begin{theorem}
Under the noiseless/noisy model introduced in Section \ref{sysmod}, there exists a strongly explicit construction and a decoding rule achieving an arbitrary but fixed $\epsilon$ average probability of error with $t = O(d \log N \log \log_d N)$ number of 
tests that can be decoded in time $O(d^3 \log N \log \log_d N)$ in the regime $d = \Omega(\log^2 N)$.
\label{thm:eff_dec}
\end{theorem}
The proof of the above theorem can be found in Appendix \ref{appendix:eff_dec}. We note that with only $\log \log_d N$ extra factor in the number of tests, the decoding complexity can be brought to the desired $O(\textnormal{poly}(t))$ complexity. We further note that the number of tests becomes order-optimal in the regime $d = \Theta(N^{\alpha})$ for some $\alpha \in (0, 1)$. In Table \ref{t1} we provide the results presented in this work along with the related results in the literature.

\begin{table*}
\begin{center}
\renewcommand{\arraystretch}{1.50}
\begin{tabular}{|c|c|c|c|}
\hline
Reference & Number of tests & Decoding complexity & Construction \\
\hline 
\cite{johnson14, chan14} & $t = \Theta(d \log N)$ & $O(t N)$ & Randomized \\
\hline
\cite{mazumdar16} & $t = O(d \log^2 N/\log d)$ & $ O(t N) $ & Strongly explicit \\
\hline
\cite{lee2016saffron} & $t =O(d \log d \log N)$ & $O(d \log d \log N)$ & Randomized \\
\hline
\cite{jaggi2017} & $t = O(d \log d \log N)$  & $O(d(\log^2 d + \log N))$ & Randomized \\
\hline
This work & $t = \Theta(d \log N)$ & $O(t N)$ & Strongly explicit \\
\hline
This work & $t = O(d \log N \log \log_d N)$ & $O(d^3 \log N \log \log_d N)$ & Strongly explicit 
\\
\hline
\end{tabular}
\end{center}
\caption{Comparison of non-adaptive probabilistic group testing results. We note that the main focus in \cite{johnson14, chan14} is the implied constant in $t = \Theta(d \log N)$.}
\label{t1}
\end{table*}


\section{Related Work}
\label{sec:related}

The literature on the non-adaptive group testing framework includes both explicit and random test designs. We refer the reader to \cite{du2000combinatorial} for a survey. In combinatorial group testing, the famous construction introduced by Kautz and Singleton \cite{kautz1964} achieves $t = O(d^2 \log_d^2 N)$ tests matching the best known lower bound $\min \{N, \Omega(d^2 \log_d N) \}$ \cite{d1982bounds, furedi1996onr} in the regime where $d = \Theta(N^{\alpha})$ for some $\alpha \in (0, 1)$. However, this strongly explicit construction is suboptimal in the regime where $d = O(\textnormal{poly}(\log N))$. An explicit construction achieving $t = O(d^2 \log N)$ was introduced by Porat and Rothschild in \cite{porat2008}. While $t = O(d^2 \log N)$ is the best known achievability result in combinatorial group testing framework, there is no strongly explicit construction matching it to the best of our knowledge. Regarding efficient decoding, recently Indyk, Ngo and Rudra \cite{indyk2010} introduced a randomized construction with $t = O(d^2 \log(N))$ tests that could be decoded in time $\textnormal{poly}(t)$. 
Furthermore, the construction in \cite{indyk2010} can be derandomized in the regime $d = O(\log N/\log \log N)$.  
Later, Ngo, Porat and Rudra \cite{ngo11} removed the constraint on $d$ and provided an explicit construction that can be decoded in time $\textnormal{poly}(t)$.
The main idea of \cite{indyk2010} was to consider \em list-disjunct \em matrices; a similar idea was considered by Cheraghchi in \cite{mahdi2009}, which obtained explicit constructions of non-adaptive group testing schemes that handle noisy tests and return a list of defectives that may include false positives.

There are various schemes relaxing the zero-error criteria in the group testing problem. For instance, the model mentioned above, where the decoder always outputs
a small super-set of the defective items, was studied in \cite{mahdi2009, rykov, bonis05, rashad90}. These constructions have efficient ($\textnormal{poly}(t)$-time) decoding algorithms, and so can be used alongside constructions without sublinear time decoding algorithms to speed up decoding.
Another framework where the goal is to recover at least a $(1 - \epsilon)$-fraction (for any arbitrarily small $\epsilon > 0$) of the defective set with high probability was studied in \cite{lee2016saffron} where the authors provided a scheme with order-optimal $O(d \log N)$ tests and the computational complexity. There are also different versions of the group testing problem in which a test can have more than two outcomes \cite{sobel71, hwang84} or can be threshold based \cite{CHEN20091581, mahdi2010, thach17}. More recently, sparse group testing frameworks for both combinatorial and probabilistic settings were studied in \cite{inan17allerton, gandikota2016sparse, inan18isit}.

When the defective set is assumed to be uniformly random, it is known that $t = \Theta(d \log N)$ is order-optimal for achieving the exact recovery of the defective set with vanishing probability of error (which is the model considered in this work) in the broad regime $d = O(N^{\alpha})$ for some $\alpha \in (0, 1)$ using random designs and information-theoretical tools \cite{ chan14, cevher16}. These results also include the noisy variants of the group testing problem. Efficient recovery algorithms with nearly optimal number of tests were introduced recently in \cite{lee2016saffron} and \cite{jaggi2017}. Regarding deterministic constructions of almost disjunct matrices, recently Mazumdar \cite{mazumdar16} introduced an analysis connecting the group testing properties with the \textit{average} Hamming distance between the columns of the measurement matrix and obtained (strongly) explicit constructions with $t = O(d \log^2 N/ \log d)$ tests. While this result is order-optimal in the regime where $d = \Theta(N^{\alpha})$ for some $\alpha \in (0, 1)$, it is suboptimal for moderate values of $d$ (e.g., $d = O(\textnormal{poly}(\log N))$). 
The performance of the Kautz-Singleton construction in the random model has been studied empirically~\cite{Erlich15}, but we are not aware of any theoretical analysis of it beyond what follows immediately from the distance of Reed-Solomon codes.
To the best of our knowledge there is no known explicit/strongly explicit construction achieving $t = \Theta(d \log N)$ tests in general for the noiseless/noisy version of the probabilistic group testing problem. 


\section{Conclusion}
\label{sec:conclusion}

In this work, we showed that the Kautz-Singleton construction is order-optimal in the noiseless and noisy variants of the probabilistic group testing problem. To the best of our knowledge, this is the first (strongly) explicit construction achieving order-optimal number of tests in the probabilistic group testing setting for poly-logarithmic (in $N$) values of $d$. We provided a novel analysis departing from the classical approaches in the literature that use combinatorial properties of the underlying code. We instead directly explored the probability of a non-defective item being covered by a random defective set using the properties of Reed-Solomon codes in our analysis. Furthermore, by using a recursive technique, we converted the Kautz-Singleton construction into a construction that is also efficiently decodable with only a log-log factor increase in number of tests which provides interesting tradeoffs compared to the existing results in the literature. 

There are a number of nontrivial extensions to our work. Firstly, it would be interesting to extend these results to the regime $d = o(\log^2 N)$. Another interesting line of work would be to find a deterministic/randomized construction achieving order-optimal $t = \Theta(d \log N)$ tests and is also efficiently decodable.

\appendix

\subsection{Proof of Theorem \ref{thm:Exp_noiseless}}
\label{appendix:noiseless}

Let $N$ be the number of items and $d$ be the size of the random defective set. We will employ the Kautz-Singleton construction which takes a $[n, k]_q$ RS code and replaces each $q$-ary symbol by a unit weight binary vector of length $q$ using identity mapping. This corresponds to mapping a symbol $i \in [q]$ to the vector in $\{0, 1\}^q$ that has a 1 in the $i$'th position and zero everywhere else (see Section \ref{sec:KS_cons} for the full description). Note that the resulting $t \times N$ binary matrix $M$ has $t = n q$ tests. We shall later see that the choice $q = 4d$ and $n = \Theta(\log N)$ is appropriate, therefore, leading to $t = \Theta(d \log N)$ tests.

We note that for any defective set the cover decoder provides an exact recovery given that none of the non-defective items are covered by the defective set. Recall that a column $M_i$ is covered by a set of columns $M_{j_1}, \ldots, M_{j_l}$ with ${j_1}, \ldots, {j_l} \in [N]$ if the support of $M_i$ is contained in the union of the supports of columns  $M_{j_1}, \ldots, M_{j_l}$. Note that in the noiseless case the measurement vector $Y$ is given by the Boolean OR of the columns corresponding to the defective items. Therefore, the measurement vector $Y$ covers all defective items, and the cover decoder can achieve exact recovery if none of the non-defective items are covered by the measurement vector $Y$ (or equivalently the defective set).

For $s \subseteq [N]$, we define $\mathcal{A}^s$ as the event that there exists a non-defective column of $M$  that is covered by the defective set $s$. Define $\mathcal{A}_i^s$ as the event that the non-defective column $M_i$ ($i \notin s$) is covered by the defective set $s$. We can bound the probability of error as follows:
\begin{align}
P_{e} & \leq \sum \limits_{s \subseteq [N], |s| = d} 1(\mathcal{A}^s) \Pr(S=s) \nonumber \\
& \leq \dfrac{1}{\binom{N}{d}} \sum \limits_{s \subseteq [N], |s| = d} \ \sum \limits_{i \in [N] \backslash s} 1(\mathcal{A}_i^s) \nonumber \\ & =
\dfrac{1}{\binom{N}{d}}  \sum \limits_{i \in [N]} \ \sum \limits_{s \subseteq [N]/\{i\}, |s| = d}  1(\mathcal{A}_i^s) \nonumber \\ & 
= \dfrac{\binom{N-1}{d}}{\binom{N}{d}} \sum \limits_{i \in [N]} \ \dfrac{1}{\binom{N-1}{d}} \sum \limits_{s \subseteq [N]/\{i\}, |s| = d} 1(\mathcal{A}_i^s) 
\nonumber  \\
& =  \dfrac{N-d}{N} \sum \limits_{i \in [N]} \Pr\left(\mathcal{A}_i^{S_{[N]/\{i\}}}\right) \label{prob_error}
\end{align}
where in the last equation $S_{[N]/\{i\}}$ is uniformly distributed on the sets of size $d$ among the items in $[N]/\{i\}$ and $1(\cdot)$ denotes the indicator function of an event.

Fix any $n$ distinct elements $\alpha_1, \alpha_2, \ldots, \alpha_n$ from $\mathbb{F}_q$. We denote $\Psi \triangleq \{\alpha_1, \alpha_2, \ldots, \alpha_n\}$. We note that due to the structure of mapping to the binary vectors in the Kautz-Singleton construction, a column $M_i$ is covered by the random defective set $S$ if and only if the corresponding symbols of $M_i$ are contained in the union of symbols of $S$ in the RS code for all rows in $[n]$. Recall that there is a $k-1$ degree polynomial $f_m(X) = \sum_{i=0}^{k-1} m_i X^i$ corresponding to each column in the RS code and the corresponding symbols in the column are the evaluation of $f_m(X)$ at  $\alpha_1, \alpha_2, \ldots, \alpha_n$. Denoting $f_{m_i}(X)$ as the polynomial corresponding to the column $M_i$, we have
\begin{align*}
\Pr\left(\mathcal{A}_i^{S_{[N]/\{i\}}}\right) & = 
\Pr\left(f_{m_i}(\alpha) \in \left\{f_{m_j}(\alpha) : j \in S_{[N]/\{i\}} \right\} \ \forall \ \alpha \in \Psi\right) \\
& =
\Pr\left(0 \in \left\{f_{m_j}(\alpha) - f_{m_i}(\alpha) : j \in S_{[N]/\{i\}} \right\} \ \forall \ \alpha \in \Psi\right). 
\end{align*}
We note that the columns of the RS code contain all possible (at most) $k-1$ degree polynomials, therefore, the set  
$\left\{f_{m_j}(\alpha) - f_{m_i}(\alpha) : j \in {[N]/\{i\}}\right\}$ is sweeping through all possible (at most)  $k-1$ degree polynomials except the zero polynomial. Therefore, the randomness of $S_{[N]/\{i\}}$ that generates the random set $\left\{f_{m_j}(\alpha) - f_{m_i}(\alpha) : j \in S_{[N]/\{i\}} \right\}$ can be translated to the random set of polynomials $\{f_{m_j}(X) : j \in S'\}$ that is generated by picking $d$ nonzero polynomials of degree (at most) $k-1$ without replacement. This gives
\begin{align*}
\Pr\left(0 \in \left\{f_{m_j}(\alpha) - f_{m_i}(\alpha) : j \in S_{[N]/\{i\}} \right\} \ \forall \ \alpha \in \Psi\right)
= \Pr\left(0 \in \left\{f_{m_j}(\alpha) : j \in S' \right\} \ \forall \ \alpha \in \Psi \right).
\end{align*}
We define the random polynomial $h(X) \triangleq \prod \limits_{j \in S'} f_{m_j}(X)$. Note that
\begin{align*}
0 \in \{f_{m_j}(\alpha) : j \in S'\} \ \forall \ \alpha \in \Psi \ \Leftrightarrow \ h(\alpha) = 0 \ \forall \ \alpha \in \Psi.
\end{align*}
We next bound the number of roots of the polynomial $h(X)$. We will use the following result from \cite{raz}. 
\begin{lemma}[{{\cite[Lemma 3.9]{raz}}}]
Let $R_q(l, k)$ denote the set of nonzero polynomials over $\mathbb{F}_q$ of degree at most $k$ that have exactly $l$ distinct roots in $\mathbb{F}_q$. For all powers $q$ and integers $l, k,$
\begin{align*}
|R_q(l, k)| \leq q^{k+1} \cdot \dfrac{1}{l!}.
\end{align*}
\label{root}
\end{lemma}
Let $r$ denote the number of roots of a random nonzero polynomial of degree at most $k-1$. One can observe that $\mathbb{E}[r] \leq 1$ by noting that there is exactly one value of $m_0$ that makes $f_m(X) = 0$ for any fixed $X$ and $m_1, \ldots, m_{k-1}$ and the inequality is due to excluding the zero polynomial. 
Furthermore, using Lemma \ref{root}, we get
\begin{align*}
\mathbb{E}[r^2] & \leq \sum \limits_{i = 1}^{k-1} \dfrac{i^2}{i!} 
\\ & = \sum \limits_{i = 1}^{k-1} \dfrac{i}{(i-1)!} 
\\ & = \sum \limits_{i = 1}^{k-1} \dfrac{i-1}{(i-1)!} + \sum \limits_{i = 1}^{k-1} \dfrac{1}{(i-1)!} 
\\ & < 2e
\end{align*}
where the first inequality is due to $\Pr(r = i) = |R_q(i, k-1)|/q^k \leq 1/i!$ from Lemma \ref{root}. Hence we can bound $\mathbb{E}[r^2] < 6$. We denote $r_i$ as the number of roots of the polynomial $f_{m_i}(X)$ and $r_h$ as the number of roots of the polynomial $h(X)$. Note that $r_h \leq \sum_{j \in S'} r_j$. We will use the following Bernstein concentration bound for sampling without replacement \cite{bardenet2015}:
\begin{proposition}[{{\cite[Proposition 1.4]{bardenet2015}}}]
Let $\mathcal{X} =\{ x_1, \ldots, x_N\}$ be a finite population of $N$ points and $X_1, \ldots, X_n$ be a random sample drawn without replacement from $\mathcal{X}$. Let $a = \min \limits_{1\leq i \leq N} x_i$ and $b = \max \limits_{1\leq i \leq N} x_i$. Then for all $\epsilon>0$, 
\[
\Pr\left( \dfrac{1}{n} \sum \limits_{i=1}^{n} X_i - \mu > \epsilon \right) \leq \exp \left(-\dfrac{n \epsilon^2}{2 \sigma^2 + (2/3)(b-a)\epsilon} \right) 
\] where $\mu = \frac{1}{N} \sum_{i=1}^{N} x_i $ is the mean of $\mathcal{X}$ and $\sigma^2 = \frac{1}{N} \sum_{i=1}^{N} (x_i - \mu)^2$ is the variance of $\mathcal{X}$.
\end{proposition}
We apply the inequality above to $\sum_{j \in S'} r_j$ and obtain
\begin{align*}
\Pr \left( \sum \limits_{j \in S'} r_j > 2d \right) & = \Pr \left( \dfrac{1}{d} \sum \limits_{j \in S'} r_j > 2 \right) \\ & \leq \Pr \left( \dfrac{1}{d} \sum \limits_{j \in S'} (r_j - \mathbb{E}[r_j]) > 1 \right) \\ & \leq \exp \left( - \dfrac{d}{12 +  k (2/3) } \right) \\ & \leq 
\exp \left( - \dfrac{d}{16 k} \right).
\end{align*}
We have $k = \log N/\log q$, hence, under the regime $d = \Omega(\log^2 N)$, the last quantity is bounded by $N^{- c \log q}$ for some constant $c > 0$. Hence the number of roots of the polynomial $h(X)$ is bounded by $2d$ with high probability.

Given the condition that the number of roots of the polynomial $h(X)$ is bounded by $2d$ and 
the random set of polynomials $\{f_{m_j}(X) : j \in S'\}$ is picked from the nonzero polynomials of degree at most $k-1$ without replacement, due to the symmetry in the position of the roots of the randomly selected polynomials, we claim that the probability of satisfying $h(\alpha) = 0 $ for all $ \alpha \in \Psi$ is bounded by the probability of covering $n$ elements from a field of size $q$ by picking $2d$ elements randomly without replacement. We next prove this claim. We define the set $R(h) \coloneqq \{ \alpha \in  \mathbb{F}_q : h(\alpha) = 0\}$
and we emphasize that this is not a multiset, i.e., the repeated roots appear as a single element. We begin with the following observation.
\begin{claim}
Let $l > 0$, and condition on the event that $\vert R(h) \vert = l$. Then $R(h)$ is uniformly distributed among all sets $\Lambda \subseteq \mathbb{F}_q$ of size $l$.
\end{claim}
\begin{proof}
For $f \in \mathbb{F}_q[X]$, we can write 
\begin{align*}
f(X) = g_f(X) \cdot \prod \limits_{\gamma_i \in R(f)} (X - \gamma_i)^{c_i},
\end{align*}
where $c_i$ is the corresponding multiplicity of the root $\gamma_i$ and $g_f \in \mathbb{F}_q[X]$ does not have any linear factor. We note that this decomposition is unique. 
For $\Lambda \subseteq \mathbb{F}_q$ of size $l$, let
\begin{align*}
H_{\Lambda} \coloneqq \left\{ \{f_1(X), \ldots,  f_d(X)\} : R\left(\prod_i f_i(X)\right) = \Lambda \right\}.
\end{align*}
Let $\Lambda' \subseteq \mathbb{F}_q$ such that $\vert \Lambda'  \vert = l$ and $\Lambda'  \neq \Lambda$. Then $\vert H_{\Lambda}  \vert = \vert H_{\Lambda'}  \vert$.
Indeed, let $\varphi : \mathbb{F}_q \rightarrow \mathbb{F}_q$ be a bijection such that $\varphi(\Lambda) = \Lambda'$. Then $\Phi : H_{\Lambda} \rightarrow H_{\Lambda'}$ given by
\begin{align*}
\Phi(f) = g_f(X) \cdot \prod \limits_{\gamma_i \in R(f)} (X - \varphi(\gamma_i))^{c_i},
\end{align*}
and $\Phi(\{f_1, \ldots, f_d\}) \coloneqq \{\Phi(f_1), \ldots, \Phi(f_d)\}$ is a bijection. 

We further note that $R(h) = \Lambda \Rightarrow \vert R(h) \vert = l$, so
\begin{align*}
\Pr \{ R(h) = \Lambda \given[\big]   \vert R(h) \vert = l \} &= \dfrac{\Pr \{ R(h) = \Lambda \}}{\Pr \{ \vert R(h) \vert = l \}} \\ &
= \dfrac{\Pr \{ \{f_1, \ldots, f_d\} \in H_{\Lambda} \}}{\Pr \{ \vert R(h) \vert = l \}} \\ &
 \stackrel{(i)}{=} \dfrac{\Pr \{ \{f_1, \ldots, f_d\} \in H_{\Lambda'} \}}{\Pr \{ \vert R(h) \vert = l \}} 
 \\ & = \Pr \{ R(h) = \Lambda' \given[\big]   \vert R(h) \vert = l \},
\end{align*}
where $(i)$ is due to $\vert H_{\Lambda}  \vert = \vert H_{\Lambda'}  \vert$ and we pick $f_1, \ldots, f_d$ uniformly without replacement.
\end{proof}

Based on this, if we ensure $n \leq 2d$, then it follows that
\begin{align*}
& \Pr \{ R(h) \supseteq \Psi \given[\big]   \vert R(h) \vert \leq 2d \} \\
& = \sum \limits_{l \leq 2d} \Pr \{ R(h) \supseteq \Psi \given[\big]   \vert R(h) \vert = l \} 
\Pr \{  \vert R(h) \vert = l \given[\big] \vert R(h) \vert \leq 2d \}
\\ & \leq
\max \limits_{n \leq l \leq 2d} \Pr \{ R(h) \supseteq \Psi \given[\big]   \vert R(h) \vert = l \} 
\\ & = \max \limits_{n \leq l \leq 2d} \dfrac{\binom{q-n}{l - n}}{\binom{q}{l}}.
\end{align*}
Let us fix $q = 4d$. We then have
\begin{align*}
\Pr \{ R(h) \supseteq \Psi \given[\big]   \vert R(h) \vert \leq 2d \}  & \leq \dfrac{\binom{4d-n}{2d - n}}{\binom{4d}{2d}} 
\\ &  =
\dfrac{(4d-n)!}{(2d-n)!(2d)!} \dfrac{(2d)! (2d)!}{(4d)!} \\ & 
= \dfrac{2d \ldots (2d - n + 1)}{4d  \ldots (4d - n + 1)} \\ & 
\leq \left( \dfrac{1}{2} \right)^n.
\end{align*}
Therefore, $\Pr(\mathcal{A}_i^S)$ is bounded by
\begin{align*}
\Pr(\mathcal{A}_i^S) & \leq \Pr \{ R(h) \supseteq \Psi \given[\big]   \vert R(h) \vert \leq 2d \} +
\Pr \{ \vert R(h) \vert > 2d \} \\ &
\leq \left( \dfrac{1}{2} \right)^n + N^{- c \log q}.
\end{align*}
Applying the summation over all $i \in [N]$ in \eqref{prob_error}, we obtain $P_e \leq N^{1 - c \log q} + N 2^{-n}$. Therefore, under the regime $d = \Omega(\log^2 N)$, the average probability of error can be bounded as $P_e \leq N^{- \Omega(\log q)} + N^{-\delta}$ by choosing $n = (1 + \delta) \log N$. The condition $n \leq 2d$ required in the proof is also satisfied under this regime. Note that the resulting $t \times N$ binary matrix $M$ has $t = n q = \Theta(d \log N)$ tests.

\subsection{Proof of Theorem \ref{thm:Exp_noisy}}
\label{appendix:noisy}

We begin with describing the decoding rule. Since we are considering the noisy model, we will slightly modify the cover decoder employed in the noiseless case. For any defective item with codeword weight $w$, in the noiseless outcome the tests in which this item participated will be all positive. On the other hand, when the noise is added, $w p$ of these tests will flip in expectation. Based on this observation (see {\bf No-CoMa} in \cite{chan14} for a more detailed discussion), we consider the following decoding rule. For any item $i \in [N]$, we first denote $w_i$ as the weight of the corresponding column $M_i$ and $\hat{w}_i$ as the number of rows $k \in [t]$ where both $M_{k, i} = 1$ and $Y_k = 1$. If $\hat{w}_i \geq w_i (1 - p(1 + \tau))$, then the $i$th item is declared as defective, else it is declared to be non-defective.

Under the aforementioned decoding rule, an error event happens either when $\hat{w}_i < w_i (1 - p(1 + \tau))$ for a defective item $i$ or $\hat{w}_i \geq w_i (1 - p(1 + \tau))$ for a non-defective item $i$. Using the union bound, we can bound the probability of error as follows:

\begin{align}
P_e
& \leq \dfrac{1}{\binom{N}{d}} \sum \limits_{s \subseteq [N], |s| = d} \bigg[ \ \sum \limits_{i \in [N] \backslash s} \Pr  \{ \textnormal{$\hat{w}_i  \geq w_i (1 - p(1 + \tau))$} \} 
+ \sum \limits_{i \in s} \Pr  \{ \textnormal{$\hat{w}_i  < w_i (1 - p(1 + \tau))$} \} \bigg] \nonumber
\\ & =
\dfrac{1}{\binom{N}{d}}  \sum \limits_{i \in [N]} \ \sum \limits_{s \subseteq [N]/\{i\}, |s| = d} \Pr  \{ \hat{w}_i  \geq w_i (1 - p(1 + \tau)) \} 
 +  \dfrac{1}{\binom{N}{d}} \sum \limits_{s \subseteq [N], |s| = d} \ \sum \limits_{i \in s} \Pr \{ \hat{w}_i  < w_i (1 - p(1 + \tau)) \} \nonumber \\ & 
= \dfrac{\binom{N-1}{d}}{\binom{N}{d}} \bigg( \sum \limits_{i \in [N]} \ \dfrac{1}{\binom{N-1}{d}} 
\sum \limits_{s \subseteq [N]/\{i\}, |s| = d} \Pr \{ \hat{w}_i  \geq w_i (1 - p(1 + \tau)) \}  \bigg) \nonumber
\\ &  \hspace{3.2in}
 + \dfrac{1}{\binom{N}{d}} \sum \limits_{s \subseteq [N], |s| = d} \ \sum \limits_{i \in s} \Pr  \{  \hat{w}_i  < w_i (1 - p(1 + \tau)) \} \nonumber \\ & 
= \dfrac{N-d}{N} \sum \limits_{i \in [N]} \Pr  \{ \hat{w}_i  \geq w_i (1 - p(1 + \tau)) \} 
+ \dfrac{1}{\binom{N}{d}} \sum \limits_{s \subseteq [N], |s| = d} \ \sum \limits_{i \in s} \Pr  
\{ \hat{w}_i  < w_i (1 - p(1 + \tau)) \} \label{pf2:noise} \\ & 
\eqqcolon P_1 + P_2, \nonumber
\end{align}
where we denote the first term of \eqref{pf2:noise} as $P_1$ and the second one as $P_2$ in the last equation. We point out that in the first term of \eqref{pf2:noise} the randomness is both due to the noise and the defective set that is uniformly distributed among the items in $[N]/\{i\}$ whereas in the second term the randomness is due to the noise.

We will employ the Kautz-Singleton construction which takes a $[n, k]_q$ RS code and replaces each $q$-ary symbol by unit weight binary vectors of length $q$ using identity mapping. This corresponds to mapping a symbol $i \in [q]$ to the vector in $\{0, 1\}^q$ that has a 1 in the $i$'th position and zero everywhere else (see Section \ref{sec:KS_cons} for the full description). Note that the resulting $t \times N$ binary matrix $M$ has $t = n q$ tests. We shall later see that the choice $q = 24d$ and $n = \Theta(\log N)$ is appropriate, therefore, leading to $t = \Theta(d \log N)$ tests. Fix any $n$ distinct elements $\alpha_1, \alpha_2, \ldots, \alpha_n$ from $\mathbb{F}_q$. We denote $\Psi \triangleq \{\alpha_1, \alpha_2, \ldots, \alpha_n\}$.

We begin with $P_2$. Fix any defective set $s$ in $[N]$ with size $d$ and fix an arbitrary element $i$ of this set. We first note that $w_i = n$ due to the structure of the Kautz-Singleton construction. We further note that before the addition of noise the noiseless outcome will have positive entries corresponding to the ones where $M_{k, i} = 1$. Therefore $\Pr \{ \hat{w}_i  < w_i (1 - p(1 + \tau)) \} $ only depends on the number of bit flips due to the noise. Using Hoeffding's inequality, we have
\begin{align*}
\Pr \{ \hat{w}_i  < w_i (1 - p(1 + \tau)) \} \leq e^{-2 n p^2 \tau^2}.
\end{align*}
Summing over the $d$ defective items $i \in s$, we get $P_2 \leq d e^{-2 n p^2 \tau^2}$. 

We continue with $P_1$. We fix an item $i \in [N]$ and note that $w_i = n$. We similarly define the random polynomial $h(X) \triangleq \prod \limits_{j \in S} f_{m_j}(X)$. Let $\mathcal{A}$ be the event of $h(X)$ having at most $2d$ number of roots. We then have
\begin{align}
\Pr \{ \hat{w}_i  \geq w_i (1 - p(1 + \tau)) \}  & =  \Pr  \{ \hat{w}_i  \geq w_i (1 - p(1 + \tau))  | \mathcal{A}  \} \Pr \{ \mathcal{A} \} 
+ \Pr  \{ \hat{w}_i  \geq w_i (1 - p(1 + \tau)) | \mathcal{A}^c \} \Pr \{ \mathcal{A}^c \} \nonumber \\ & \leq 
\Pr  \{ \hat{w}_i  \geq w_i (1 - p(1 + \tau)) | \mathcal{A} \} + \Pr \{ \mathcal{A}^c \}.
\label{x1}
\end{align}
Following similar steps as in the proof of Theorem \ref{thm:Exp_noiseless} we obtain $\Pr \{ \mathcal{A}^c \} \leq N^{- c \log q}$ for some constant $c > 0$ in the regime $d = \Omega(\log^2 N)$. 

We next bound the first term in \eqref{x1}. We choose $q = 24d$ and define the random set $\Upsilon = \{ \alpha \in \Psi : f_{m_i}(\alpha) \in \{f_{m_j}(\alpha) : j \in S\} \}$. We then have
\begin{align*}
\Pr  \{ \hat{w}_i  \geq w_i (1 - p(1 + \tau))| \mathcal{A} \}  & = 
\Pr  \{ \hat{w}_i  \geq w_i (1 - p(1 + \tau)) | \mathcal{A}, |\Upsilon| \leq n/4 \} \Pr \{ |\Upsilon| \leq n/4 | \mathcal{A} \} 
\\ & \hspace{1in}
+ \Pr \{ \hat{w}_i  \geq w_i (1 - p(1 + \tau)) | \mathcal{A}, |\Upsilon| > n/4 \} \Pr \{ |\Upsilon| > n/4 | \mathcal{A} \} 
\\ & \leq 
\Pr \{ \hat{w}_i  \geq w_i (1 - p(1 + \tau)) | \mathcal{A}, |\Upsilon| \leq n/4 \}  + \Pr \{ |\Upsilon| > n/4 | \mathcal{A} \}.
\end{align*}
Let us first bound the second term $\Pr \{ |\Upsilon| > n/4 | \mathcal{A} \}$. We note that
\begin{align*}
|\Upsilon| & = |\{ \alpha \in \Psi : f_{m_i}(\alpha) \in \{f_{m_j}(\alpha) : j \in S\} \}| \\ &
= |  \{ \alpha \in \Psi : 0 \in \{f_{m_j}(\alpha) - f_{m_i}(\alpha) : j \in S\} \} | \\ &
= |\{ \alpha \in \Psi : 0 \in \{f_{m_j}(\alpha) : j \in S'\} \}|
\end{align*}
where in the last equality the random set of polynomials $\{f_{m_j}(X) : j \in S'\}$ is generated by picking $d$ nonzero polynomials of degree at most $k-1$ without replacement. This holds since $i \notin S$ and 
the columns of the RS code contain all possible (at most) $k-1$ degree polynomials, therefore, the randomness of $\{f_{m_j}(\alpha) - f_{m_i}(\alpha) : j \in S\}$ can be translated to the random set of polynomials $\{f_{m_j}(X) : j \in S'\}$ that is generated by picking $d$ nonzero polynomials of degree (at most) $k-1$ without replacement. Following similar steps of the proof of Theorem \ref{thm:Exp_noiseless} we can bound $\Pr \{ |\Upsilon| > n/4 | \mathcal{A} \}$ by considering the probability of having at least $n/4$ symbols from $\Psi$ when we pick $2d$ symbols from $[q]$ uniformly at random without replacement. Hence, if we ensure $n \leq 8d$, then we have
\begin{align*}
\Pr \{ |\Upsilon| > n/4 | \mathcal{A} \} & \leq \dfrac{\binom{n}{n/4} \binom{q-n/4}{2d - n/4}}{\binom{q}{2d}} \\ & = \dfrac{\binom{n}{n/4} \binom{24d-n/4}{2d - n/4}}{\binom{24d}{2d}} \\ & \leq (4e)^{n/4}
\dfrac{(24d-n/4)!}{(2d-n/4)!(22d)!} \dfrac{(2d)! (22d)!}{(24d)!}
\\ & = (4e)^{n/4} \dfrac{2d (2d-1) \ldots (2d - n/4 + 1)}{24d (24d - 1) \ldots (24d - n/4 + 1)} \\ & \leq \left( \dfrac{4e}{12} \right)^{n/4}
\end{align*}
where we use $\binom{n}{k} \leq (en/k)^k$ in the second inequality.

We continue with $\Pr \{ \hat{w}_i  \geq w_i (1 - p(1 + \tau)) | \mathcal{A}, |\Upsilon| \leq n/4 \}$. Note that $w_i = n$. We further note that
\begin{align*}
\mathbb{E}[\hat{w}_i] = \mathbb{E}[  \mathbb{E}[ \hat{w}_i | \Upsilon]] = 
 \mathbb{E}[ |\Upsilon| ] (1 - p) + (n -  \mathbb{E}[ |\Upsilon| ]) p.
\end{align*}
Since $p \in (0, 0.5)$ we have $\mathbb{E}[\hat{w}_i \ | \ |\Upsilon| \leq n/4] \leq (n/4) (1 - p) + (3n/4) p = n/4 + (n/2) p$. Using Hoeffding's inequality, we have
\begin{align*}
\Pr \{ \hat{w}_i  \geq w_i (1 - p(1 + \tau)) | \mathcal{A}, |\Upsilon| \leq n/4 \} &  
\leq \Pr \{\hat{w}_i - \mathbb{E}[\hat{w}_i]  \geq n(3/4 - 3p/2 - p\tau) | \mathcal{A}, |\Upsilon| \leq n/4 \} 
\\ & \leq 
e^{-2n (3/4 - 3p/2 - p \tau)^2}
\end{align*}
where the condition $3/4 - 3p/2 - p \tau > 0$ or $\tau < (3/4 - 3p/2)/p$ can be satisfied with our choice of free parameter $\tau$ since $p < 1/2$. Combining everything, we obtain
\begin{align*}
P_e & \leq N^{1 - c \log q} + N (e/3)^{n/4} + N  e^{-2n (3/4 - 3p/2 - p \tau)^2} 
+ d e^{-2 n p^2 \tau^2} \\ & \leq
N^{1 - c \log q} + N (e/3)^{n/4} + N  e^{-2n (3/4 - 3p/2 - p \tau)^2} 
+ N e^{-2 n p^2 \tau^2} \\ & =
N^{-\Omega(\log q)} + e^{\log N - \log(3/e) n/4} + 2 N e^{\log N - 9/8 (0.5 - p)^2 n}
\end{align*}
where in the last step we pick $\tau = \frac{3(0.5-p)}{4p}$. Therefore, under the regime $d = \Omega(\log^2 N)$, the average probability of error can be bounded as $P_e \leq N^{- \Omega(\log q)} + 3N^{-\delta}$ by choosing $n = \max \{ \frac{4}{\log(3/e)}, \frac{8}{9(0.5-p)^2}\}(1 + \delta) \log N$. The condition $n \leq 8d$ required in the proof is also satisfied under this regime. Note that the resulting $t \times N$ binary matrix $M$ has $t = n q = \Theta(d \log N)$ tests.

\subsection{Proof of Theorem \ref{thm:eff_dec}}
\label{appendix:eff_dec}

We begin with the noiseless case. We will use a recursive approach to obtain an efficiently decodable group testing matrix. Let $M^{\textnormal{ED}}_n$ denote such a matrix with $n$ columns in the recursion and $M^{\textnormal{KS}}_n$ denote the matrix with $n$ columns obtained by the Kautz-Singleton construction. Note that the final matrix is $M^{\textnormal{ED}}_N$.  Let $t^{\textnormal{ED}}(d, n, \epsilon)$ and $t^{\textnormal{KS}}(d, n, \epsilon)$ denote the number of tests for $M^{\textnormal{ED}}_n$ and $M^{\textnormal{KS}}_n$ respectively to detect at most $d$ defectives among $n$ columns with average probability of error $\epsilon$. We further define $D^{\textnormal{ED}}(d, n, \epsilon)$ to be the decoding time for $M^{\textnormal{ED}}_n$ with $t^{\textnormal{ED}}(d, n, \epsilon)$ rows.

We first consider the case $N = d^{2^i}$ for some non-negative integer $i$. The base case is $i = 0$, i.e., $N = d$ for which we can use individual testing and have $t^{\textnormal{ED}}(d, d, \epsilon) = d$ and $D^{\textnormal{ED}}(d, d, \epsilon) = O(d)$. For $i > 0$, we use $t^{\textnormal{ED}}(d, \sqrt{N}, \epsilon/4) \times \sqrt{N}$ matrix $M^{\textnormal{ED}}_{\sqrt{N}}$ to construct two $t^{\textnormal{ED}}(d, \sqrt{N}, \epsilon/4) \times N$ matrices $M^{(F)}$ and $M^{(L)}$ as follows. 
The $j$th column of $M^{\textnormal{ED}}_{\sqrt{N}}$ for $j \in [\sqrt{N}]$ is identical to all $i$th columns of $M^{(F)}$ for $i \in [N]$ if the \textit{first} $\frac{1}{2} \log N$ bits of $i$ is $j$ where $i$ and $j$ are considered as their respective binary representations. Similarly, the $j$th column of $M^{\textnormal{ED}}_{\sqrt{N}}$ for $j \in [\sqrt{N}]$ is identical to all $i$th columns of $M^{(L)}$ for $i \in [N]$ if the \textit{last} $\frac{1}{2} \log N$ bits of $i$ is $j$. We finally construct $M^{\textnormal{KS}}_N$ that achieves $\epsilon/2$ average probability of error and stack $M^{(F)}$, $M^{(L)}$, and  $M^{\textnormal{KS}}_N$ to obtain the final matrix $M^{\textnormal{ED}}_N$. Note that, this construction gives us the following recursion in terms of the number of tests
\begin{align*}
t^{\textnormal{ED}}(d, N, \epsilon) = 2 t^{\textnormal{ED}}(d, \sqrt{N}, \epsilon/4) + t^{\textnormal{KS}}(d, N, \epsilon/2).
\end{align*}
When $N = d^{2^i}$, note that $2^i = \log_d N$ and $i = \log \log_d N$. To solve for $t^{\textnormal{ED}}(d, d^{2^i}, \epsilon)$, we iterate the recursion as follows.
\begin{align}
t^{\textnormal{ED}}(d, d^{2^i}, \epsilon) & = 2 t^{\textnormal{ED}}(d, d^{2^{i-1}}, \epsilon/4) + t^{\textnormal{KS}}(d, d^{2^i}, \epsilon/2) \nonumber \\ &
=  4 t^{\textnormal{ED}}(d, d^{2^{i-2}}, \epsilon/16) + 2 t^{\textnormal{KS}}(d, d^{2^{i-1}}, \epsilon/8) + t^{\textnormal{KS}}(d, d^{2^i}, \epsilon/2) \nonumber \\ & \ \ \vdots 
\nonumber \\ & = 2^i t^{\textnormal{ED}}(d, d, \epsilon/2^{2i}) + \sum \limits_{j=0}^{i-1} 2^j  t^{\textnormal{KS}}(d, d^{2^{i-j}}, \epsilon/2^{j+1}) \nonumber \\ & =
2^i \cdot d + \sum \limits_{j=0}^{i-1} 2^j \cdot 4 d \log\left(d^{2^{i-j}}/\left(\epsilon/2^{j+1}\right)\right)
\label{ks_eps} \\ & = 2^i \cdot d + \sum \limits_{j=0}^{i-1} 2^j \cdot 4 d \left( 2^{i-j} \log d + (j+1) \log 2 + \log(1/\epsilon) \right) \nonumber \\ & \leq
2^i \cdot d + i \cdot 2^i \cdot 4d \log d + 4  d \sum \limits_{j=0}^{i-1} 2^j (j+1) +  2^i \cdot 4d \log(1/\epsilon) \nonumber \\ & \leq 
2^i \cdot d + i \cdot 2^i \cdot 4d \log d + i \cdot 2^i \cdot 4  d +  2^i \cdot 4d \log(1/\epsilon) \label{ks_eps_last}
\end{align}
where in \eqref{ks_eps} for simplicity we ignore the term $N^{- \Omega(\log q)}$ in the probability of error for Theorem \ref{thm:Exp_noiseless} and take $t^{\textnormal{KS}}(d, N, \epsilon) = 4 d \log N/\epsilon$. Replacing $2^i = \log_d N$ and $i = \log \log_d N$ in \eqref{ks_eps_last}, it follows that 
\begin{align*}
t^{\textnormal{ED}}(d, N, \epsilon) & = O\left(d \log N \log \log_d N  
+   d \log_d N \log\left( (\log_d N)/\epsilon \right)\right).
\end{align*}
Note that this gives $t^{\textnormal{ED}}(d, N) = O(d \log N \log \log_d N)$ in the case where $\epsilon = \Theta(1)$.

In the more general case, let $i \geq 1$ be the smallest integer such that $d^{2^{i-1}} < N \leq d^{2^i}$. It follows that $i < \log \log_d N + 1$. We can construct $M^{\textnormal{ED}}_N$ from $M^{\textnormal{ED}}_{d^{2^i}}$ by removing its last $d^{2^i} - N$ columns. We can operate on $M^{\textnormal{ED}}_N$ as if the removed columns were all defective. Therefore the number of tests satisfies $t^{\textnormal{ED}}(d, N) = O(d \log N \log \log_d N)$. 

We next describe the decoding process. We run the decoding algorithm for $M^{\textnormal{ED}}_{\sqrt{N}}$ with the components of the outcome vector $Y$ corresponding to $M^{(F)}$ and $M^{(L)}$ to compute the estimate sets $\hat{S}^{(F)}$ and $\hat{S}^{(L)}$. By induction and the union bound, the set $S' = \hat{S}^{(F)} \times \hat{S}^{(L)}$ contains all the indices $i \in S$ with error probability at most $\epsilon/2$. We further note that $\vert S' \vert \leq d^2$. We finally apply the naive cover decoder to the component of $M^{\textnormal{ED}}_{N}$ corresponding to $M^{\textnormal{KS}}_{N}$ over the set $S'$ to compute the final estimate $\hat{S}$ which can be done with an additional $O(d^2 \cdot t^{\textnormal{KS}}(d, N, \epsilon/2))$ time. By the union bound overall probability of error is bounded by $\epsilon$. This decoding procedure gives us the following recursion in terms of the decoding complexity
\begin{align*}
D^{\textnormal{ED}}(d, N, \epsilon) = 2 D^{\textnormal{ED}}(d, \sqrt{N}, \epsilon/4) + O(d^2 \cdot t^{\textnormal{KS}}(d, N, \epsilon/2)).
\end{align*}
When $N = d^{2^i}$, to solve for $D^{\textnormal{ED}}(d, d^{2^i}, \epsilon)$, we iterate the recursion as follows.
\begin{align}
D^{\textnormal{ED}}(d, d^{2^i}, \epsilon) & = 2 D^{\textnormal{ED}}(d, d^{2^{i-1}}, \epsilon/4) + c \cdot d^2 \cdot t^{\textnormal{KS}}(d, d^{2^i}, \epsilon/2) \nonumber \\ &
=  4 D^{\textnormal{ED}}(d, d^{2^{i-2}}, \epsilon/16) + 2c \cdot d^2 \cdot t^{\textnormal{KS}}(d, d^{2^{i-1}}, \epsilon/8) + c \cdot d^2 \cdot t^{\textnormal{KS}}(d, d^{2^i}, \epsilon/2) \nonumber \\ & \ \ \vdots 
\nonumber \\ & = 2^i D^{\textnormal{ED}}(d, d, \epsilon/2^{2i}) + \sum \limits_{j=0}^{i-1} 2^j c \cdot d^2 \cdot  t^{\textnormal{KS}}(d, d^{2^{i-j}}, \epsilon/2^{j+1}) \nonumber \\ & =
2^i \cdot O(d) + \sum \limits_{j=0}^{i-1} 2^j c \cdot 4 d^3 \log\left(d^{2^{i-j}}/\left(\epsilon/2^{j+1}\right)\right)
\nonumber  \\ & \leq
2^i \cdot O(d) + i \cdot 2^i \cdot 4cd^3 \log d + i \cdot 2^i \cdot 4 c d^3 +  2^i \cdot 4cd^3 \log(1/\epsilon) \label{ks_dec_last}
\end{align}
where \eqref{ks_dec_last} is obtained in the same way as \eqref{ks_eps_last}. Replacing $2^i = \log_d N$ and $i = \log \log_d N$ in \eqref{ks_dec_last}, it follows that
\begin{align*}
D^{\textnormal{ED}}(d, N, \epsilon) & = O\left(d^3 \log N \log \log_d N 
+  d^3 \log_d N \log\left( (\log_d N)/\epsilon \right)\right).
\end{align*}
Note that this gives $D^{\textnormal{ED}}(d, N) = O(d^3 \log N \log \log_d N)$ in the case where $\epsilon = \Theta(1)$.

The noisy case follows similar lines except the difference is that in the base case where $N = d$, we cannot use individual testing due to the noise. In this case we can do individual testing with repetitions which requires $t^{\textnormal{ED}}(d, d, \epsilon) = O(d \log(d/\epsilon))$ and $D^{\textnormal{ED}}(d, d, \epsilon) = O(d \log(d/\epsilon))$. We can proceed similarly as in the noiseless case and show that $t^{\textnormal{ED}}(d, N) = O(d \log N \log \log_d N)$ and $D^{\textnormal{ED}}(d, N) = O(d^3 \log N \log \log_d N)$.

\section*{Acknowledgements}
The third author would like to thank Atri Rudra and Hung Ngo for helpful conversations. We thank the anonymous reviewers for helpful comments and suggestions.

\bibliographystyle{IEEEbib}
\bibliography{references}

\begin{thebibliography}{10}

\bibitem{dorfman1943detection}
R.~Dorfman,
\newblock ``The detection of defective members of large populations,''
\newblock {\em The Annals of Mathematical Statistics}, vol. 14, no. 4, pp.
  436--440, 1943.

\bibitem{chen2008decoding}
H-B. Chen and F.~K. Hwang,
\newblock ``A survey on nonadaptive group testing algorithms through the angle
  of decoding,''
\newblock {\em Journal of Combinatorial Optimization}, vol. 15, no. 1, pp.
  49--59, 2008.

\bibitem{ganesan2017learning}
A.~Ganesan, S.~Jaggi, and V.~Saligrama,
\newblock ``Learning immune-defectives graph through group tests,''
\newblock {\em IEEE Transactions on Information Theory}, 2017.

\bibitem{malioutov2013exact}
D.~Malioutov and K.~Varshney,
\newblock ``Exact rule learning via boolean compressed sensing,''
\newblock in {\em International Conference on Machine Learning}, 2013, pp.
  765--773.

\bibitem{gilbert2008group}
A.~C. Gilbert, M.~A. Iwen, and M.~J. Strauss,
\newblock ``Group testing and sparse signal recovery,''
\newblock in {\em Signals, Systems and Computers, 2008 42nd Asilomar Conference
  on}. IEEE, 2008, pp. 1059--1063.

\bibitem{emad2014poisson}
A.~Emad and O.~Milenkovic,
\newblock ``Poisson group testing: A probabilistic model for nonadaptive
  streaming boolean compressed sensing,''
\newblock in {\em Acoustics, Speech and Signal Processing (ICASSP), 2014 IEEE
  International Conference on}. IEEE, 2014, pp. 3335--3339.

\bibitem{Berger1984}
T.~Berger, N.~Mehravari, D.~Towsley, and J.~Wolf,
\newblock ``Random multiple-access communication and group testing,''
\newblock {\em Communications, IEEE Transactions on}, vol. 32, no. 7, pp. 769
  -- 779, jul 1984.

\bibitem{wolf1985born}
J.~K. Wolf,
\newblock ``Born again group testing: Multiaccess communications,''
\newblock {\em Information Theory, IEEE Transactions on}, vol. 31, no. 2, pp.
  185--191, 1985.

\bibitem{luo2008neighbor}
J.~Luo and D.~Guo,
\newblock ``Neighbor discovery in wireless ad hoc networks based on group
  testing,''
\newblock in {\em Communication, Control, and Computing, 2008 46th Annual
  Allerton Conference on}. IEEE, 2008, pp. 791--797.

\bibitem{Fletcher2009}
A.~K. Fletcher, V.K. Goyal, and S.~Rangan,
\newblock ``A sparsity detection framework for on-off random access channels,''
\newblock in {\em Information Theory, 2009. ISIT 2009. IEEE International
  Symposium on}. IEEE, 2009, pp. 169--173.

\bibitem{ngo2000survey}
H.~Q. Ngo and D-Z. Du,
\newblock ``A survey on combinatorial group testing algorithms with
  applications to dna library screening,''
\newblock {\em Discrete mathematical problems with medical applications}, vol.
  55, pp. 171--182, 2000.

\bibitem{du2000combinatorial}
D.~Du and F.~K. Hwang,
\newblock {\em Combinatorial group testing and its applications}, vol.~12,
\newblock World Scientific, 2000.

\bibitem{atia2009}
G.~K. Atia and V.~Saligrama,
\newblock ``Boolean compressed sensing and noisy group testing,''
\newblock {\em Information Theory, IEEE Transactions on}, vol. 58, no. 3, pp.
  1880--1901, 2012.

\bibitem{sejdinovic2010}
D.~Sejdinovic and O.~Johnson,
\newblock ``Note on noisy group testing: Asymptotic bounds and belief
  propagation reconstruction,''
\newblock {\em CoRR}, vol. abs/1010.2441, 2010.

\bibitem{chan2001probabilisticgt}
C.~L. Chan, P.~H. Che, S.~Jaggi, and V.~Saligrama,
\newblock ``Non-adaptive probabilistic group testing with noisy measurements:
  Near-optimal bounds with efficient algorithms,''
\newblock in {\em 2011 49th Annual Allerton Conference on Communication,
  Control, and Computing (Allerton)}, Sept 2011, pp. 1832--1839.

\bibitem{cevher16}
J.~Scarlett and V.~Cevher,
\newblock ``Phase {T}ransitions in {G}roup {T}esting,''
\newblock in {\em {ACM}-{SIAM} {S}ymposium on {D}iscrete {A}lgorithms
  ({SODA})}, 2016.

\bibitem{johnson14}
M.~Aldridge, L.~Baldassini, and O.~Johnson,
\newblock ``Group testing algorithms: Bounds and simulations,''
\newblock {\em IEEE Transactions on Information Theory}, vol. 60, no. 6, pp.
  3671--3687, June 2014.

\bibitem{scarlett18}
O.~Johnson, M.~Aldridge, and J.~Scarlett,
\newblock ``Performance of group testing algorithms with near-constant tests
  per item,''
\newblock {\em IEEE Transactions on Information Theory}, vol. 65, no. 2, pp.
  707--723, Feb 2019.

\bibitem{scarlett_cevher_17}
J.~Scarlett and V.~Cevher,
\newblock ``Near-optimal noisy group testing via separate decoding of items,''
\newblock {\em IEEE Journal of Selected Topics in Signal Processing}, vol. 12,
  no. 5, pp. 902--915, Oct 2018.

\bibitem{johnson_noisy_18}
J.~{Scarlett} and O.~{Johnson},
\newblock ``{Noisy Non-Adaptive Group Testing: A (Near-)Definite Defectives
  Approach},''
\newblock {\em ArXiv e-prints}, Aug. 2018.

\bibitem{d1982bounds}
A.~G. D'yachkov and V.~V. Rykov,
\newblock ``Bounds on the length of disjunctive codes,''
\newblock {\em Problemy Peredachi Informatsii}, vol. 18, no. 3, pp. 7--13,
  1982.

\bibitem{furedi1996onr}
Z.~Furedi,
\newblock ``On r-cover-free families,''
\newblock {\em Journal of Combinatorial Theory, Series A}, vol. 73, no. 1, pp.
  172--173, 1996.

\bibitem{kautz1964}
W.~Kautz and R.~Singleton,
\newblock ``Nonrandom binary superimposed codes,''
\newblock {\em IEEE Transactions on Information Theory}, vol. 10, no. 4, pp.
  363--377, October 1964.

\bibitem{porat2008}
E.~Porat and A.~Rothschild,
\newblock ``Explicit non-adaptive combinatorial group testing schemes,''
\newblock {\em Automata, Languages and Programming}, pp. 748--759, 2008.

\bibitem{mazumdar16}
A.~Mazumdar,
\newblock ``Nonadaptive group testing with random set of defectives,''
\newblock {\em IEEE Transactions on Information Theory}, vol. 62, no. 12, pp.
  7522--7531, Dec 2016.

\bibitem{ngo11}
H.~Q. Ngo, E.~Porat, and A.~Rudra,
\newblock ``Efficiently decodable error-correcting list disjunct matrices and
  applications,''
\newblock in {\em Automata, Languages and Programming}, Berlin, Heidelberg,
  2011, pp. 557--568, Springer Berlin Heidelberg.

\bibitem{macula04}
A.~J. Macula, V.~V. Rykov, and S.~Yekhanin,
\newblock ``Trivial two-stage group testing for complexes using almost disjunct
  matrices,''
\newblock {\em Discrete Applied Mathematics}, vol. 137, no. 1, pp. 97 -- 107,
  2004.

\bibitem{malyutov78}
M.~B. Malyutov,
\newblock ``The separating property of random matrices,''
\newblock {\em Math. Notes}, vol. 23, no. 1, pp. 84–--91, 1978.

\bibitem{Zhigljavsky03}
A.~Zhigljavsky,
\newblock ``Probabilistic existence theorems in group testing,''
\newblock {\em J. Statist. Planning Inference}, vol. 115, no. 1, pp. 1--43,
  2003.

\bibitem{RS}
I.~S. Reed and G.~Solomon,
\newblock ``Polynomial codes over certain finite fields,''
\newblock {\em Journal of the Society for Industrial and Applied Mathematics},
  vol. 8, no. 2, pp. 300--304, 1960.

\bibitem{Erlich15_code}
Y.~Erlich,
\newblock ``Combinatorial pooling using {R}{S} codes,''
  \url{https://github.com/TeamErlich/pooling}, 2017.

\bibitem{Erlich15}
Y.~Erlich, A.~Gilbert, H.~Ngo, A.~Rudra, N.~Thierry-Mieg, M.~Wootters,
  D.~Zielinski, and O.~Zuk,
\newblock ``Biological screens from linear codes: theory and tools,''
\newblock {\em bioRxiv}, 2015.

\bibitem{mahdi2009}
M.~Cheraghchi,
\newblock ``Noise-resilient group testing: Limitations and constructions,''
\newblock in {\em Fundamentals of Computation Theory}, Berlin, Heidelberg,
  2009, pp. 62--73, Springer Berlin Heidelberg.

\bibitem{indyk2010}
P.~Indyk, H.~Q. Ngo, and A.~Rudra,
\newblock ``Efficiently decodable non-adaptive group testing,''
\newblock in {\em Proceedings of the Twenty-First Annual ACM-SIAM Symposium on
  Discrete Algorithms}. Society for Industrial and Applied Mathematics, 2010,
  pp. 1126--1142.

\bibitem{lee2016saffron}
K.~Lee, R.~Pedarsani, and K.~Ramchandran,
\newblock ``Saffron: A fast, efficient, and robust framework for group testing
  based on sparse-graph codes,''
\newblock in {\em Information Theory (ISIT), 2016 IEEE International Symposium
  on}. IEEE, 2016, pp. 2873--2877.

\bibitem{jaggi2017}
S.~Cai, M.~Jahangoshahi, M.~Bakshi, and S.~Jaggi,
\newblock ``Efficient algorithms for noisy group testing,''
\newblock {\em IEEE Transactions on Information Theory}, vol. 63, no. 4, pp.
  2113--2136, April 2017.

\bibitem{chan14}
C.~L. Chan, S.~Jaggi, V.~Saligrama, and S.~Agnihotri,
\newblock ``Non-adaptive group testing: Explicit bounds and novel algorithms,''
\newblock {\em IEEE Transactions on Information Theory}, vol. 60, no. 5, pp.
  3019--3035, May 2014.

\bibitem{rykov}
A.~G. D'yachkov and V.~V. Rykov,
\newblock ``A survey of superimposed code theory,''
\newblock {\em Problems Control Inform. Theory/Problemy Upravlen. Teor.
  Inform.}, vol. 12, pp. 229--242, 1983.

\bibitem{bonis05}
A.~De Bonis., L.~Gasieniec, and U.~Vaccaro,
\newblock ``Optimal two-stage algorithms for group testing problems,''
\newblock {\em SIAM J. Comput.}, vol. 34, no. 5, pp. 1253--1270, 2005.

\bibitem{rashad90}
A.~M. Rashad,
\newblock ``Random coding bounds on the rate for list-decoding superimposed
  codes,''
\newblock {\em Problems Control Inform. Theory/Problemy Upravlen. Teor.
  Inform.}, vol. 19(2), pp. 141--149, 1990.

\bibitem{sobel71}
M.~Sobel, S.~Kumar, and S.~Blumenthal,
\newblock ``Symmetric binomial group-testing with three outcomes,''
\newblock in {\em Proc. Purdue Symp. Decision Procedure}, 1971, pp. 119--160.

\bibitem{hwang84}
F.~Hwang,
\newblock ``Three versions of a group testing game,''
\newblock {\em SIAM Journal on Algebraic Discrete Methods}, vol. 5, no. 2, pp.
  145--153, 1984.

\bibitem{CHEN20091581}
H-B. Chen and H-L. Fu,
\newblock ``Nonadaptive algorithms for threshold group testing,''
\newblock {\em Discrete Applied Mathematics}, vol. 157, no. 7, pp. 1581 --
  1585, 2009.

\bibitem{mahdi2010}
M.~Cheraghchi,
\newblock ``Improved constructions for non-adaptive threshold group testing,''
\newblock {\em CoRR}, vol. abs/1002.2244, 2010.

\bibitem{thach17}
T.~V. Bui, M.~Kuribayashi, M.~Cheraghchi, and I.~Echizen,
\newblock ``Efficiently decodable non-adaptive threshold group testing,''
\newblock {\em CoRR}, vol. abs/1712.07509, 2017.

\bibitem{inan17allerton}
H.~A. Inan, P.~Kairouz, and A.~Ozgur,
\newblock ``Sparse group testing codes for low-energy massive random access,''
\newblock in {\em 2017 55th Annual Allerton Conference on Communication,
  Control, and Computing (Allerton)}, Oct 2017, pp. 658--665.

\bibitem{gandikota2016sparse}
V.~Gandikota, E.~Grigorescu, S.~Jaggi, and S.~Zhou,
\newblock ``Nearly optimal sparse group testing,''
\newblock in {\em 2016 54th Annual Allerton Conference on Communication,
  Control, and Computing (Allerton)}, Sept 2016, pp. 401--408.

\bibitem{inan18isit}
H.~A. Inan, P.~Kairouz, and A.~Ozgur,
\newblock ``Energy-limited massive random access via noisy group testing,''
\newblock in {\em Information Theory (ISIT), 2018 IEEE International Symposium
  on}. IEEE, 2018, pp. 1101--1105.

\bibitem{raz}
T.~Hartman and R.~Raz,
\newblock ``On the distribution of the number of roots of polynomials and
  explicit weak designs,''
\newblock {\em Random Structures \& Algorithms}, vol. 23, no. 3, pp. 235--263,
  2003.

\bibitem{bardenet2015}
R.~Bardenet and O-A. Maillard,
\newblock ``Concentration inequalities for sampling without replacement,''
\newblock {\em Bernoulli}, vol. 21, no. 3, pp. 1361--1385, 08 2015.

\end{thebibliography}
\end{document}